\documentclass[11pt]{article}
\usepackage{fullpage}
\usepackage{amsthm,amsmath,amssymb}
\usepackage{xcolor}
\usepackage{algpseudocode}
\usepackage{algorithm}
\usepackage{graphicx}
\usepackage{subcaption}
\usepackage{tikz}
\newsavebox{\theorembox}
\newsavebox{\factbox}
\newsavebox{\lemmabox}
\newsavebox{\corollarybox}
\newsavebox{\propositionbox}
\newsavebox{\examplebox}
\newsavebox{\conjecturebox}
\newsavebox{\algbox}
\newsavebox{\qbox}
\newsavebox{\problembox}
\newsavebox{\definitionbox}
\newsavebox{\assumptionbox}
\newsavebox{\hypothesisbox}
\savebox{\theorembox}{\noindent\bf Theorem}
\savebox{\factbox}{\noindent\bf Fact}
\savebox{\lemmabox}{\noindent\bf Lemma}
\savebox{\corollarybox}{\noindent\bf Corollary}
\savebox{\propositionbox}{\noindent\bf Proposition}
\savebox{\examplebox}{\noindent\bf Example}
\savebox{\conjecturebox}{\noindent\bf Conjecture}
\savebox{\algbox}{\noindent\bf Algorithm}
\savebox{\qbox}{\noindent\bf Question}
\savebox{\definitionbox}{\noindent\bf Definition}
\savebox{\problembox}{\noindent\bf Problem}
\savebox{\assumptionbox}{\noindent\bf Assumption}
\savebox{\hypothesisbox}{\noindent\bf Hypothesis}
\newtheorem{theorem}{\usebox{\theorembox}}
\newtheorem{lemma}[theorem]{\usebox{\lemmabox}}

\newtheorem{proposition}[theorem]{\usebox{\propositionbox}}

\newtheorem{definition}{\usebox{\definitionbox}}
\newtheorem{fact}{\usebox{\factbox}}

\newcommand{\bx}{\mathbf{x}}

\newcommand{\bw}{\mathbf{w}}
\newcommand{\bv}{\mathbf{v}}
\newcommand{\bz}{\mathbf{z}}
\newcommand{\bb}{\mathbf{b}}
\newcommand{\ba}{\mathbf{a}}

\newcommand{\cS}{\mathcal{S}}

\newcommand{\cD}{\mathcal{D}}
\newcommand{\poly}{\text{poly}}
\newcommand{\cC}{\mathcal{C}}
\newcommand{\cX}{\mathcal{X}}
\newcommand{\cY}{\mathcal{Y}}
\newcommand{\cF}{\mathcal{F}}
\newcommand{\cR}{\mathcal{R}}
\newcommand{\cL}{\mathcal{L}}
\newcommand{\cW}{\mathcal{W}}
\newcommand{\cA}{\mathcal{A}}
\newcommand{\hcL}{\hat{\cL}}
\newcommand{\relu}{\textsc{relu}}

\newcommand{\E}{\mathbb{E}}
\newcommand{\mmcs}{\textsc{MMCS}}

\newcommand{\cB}{\mathcal{B}}

\DeclareMathOperator{\opt}{OPT}
\newcommand{\optmmcs}{\opt_{\mmcs}}

\allowdisplaybreaks

\begin{document}
\title{The Computational Complexity of Training ReLU(s)}

\author{Pasin Manurangsi\thanks{Email: pasin@berkeley.edu. Supported by NSF under Grants No. CCF 1655215 and CCF 1815434.}\\
UC Berkeley
\and Daniel Reichman\\
Princeton University
}

\maketitle

\begin{abstract}
We consider the computational complexity of training depth-2 neural networks composed of rectified linear units (ReLUs).
We show that, even for the case of a single ReLU, finding a set of weights that minimizes the squared error (even approximately) for a given training set is NP-hard. We also show that for a simple network consisting of two ReLUs, the error minimization problem is NP-hard, even in the realizable case. We complement these hardness results by showing that, when the weights and samples belong to the unit ball, one can (agnostically) \emph{properly} and reliably learn depth-2 ReLUs with $k$ units and error at most $\epsilon$ in time $2^{(k/\epsilon)^{O(1)}}n^{O(1)}$; this extends upon a previous work of Goel et al.~\cite{goel2016reliably} which provided efficient \emph{improper} learning algorithms for ReLUs.
\end{abstract}

A \emph{rectifier} is the real function $[z]_{+} := \max(0,z)$.
A rectified linear unit (ReLU) is a function  $f(\bz):\mathbb{R}^n\rightarrow \mathbb{R}$ of the form
$f(\bz)= [\langle\bw,\bz\rangle+b]_+$ where $\bw \in \mathbb{R}^n$ and $b \in \mathbb{R}$ are fixed.
A depth-2 neural network $f$ with $k$ ReLU units is a function $f:\mathbb{R}^n\rightarrow \mathbb{R}$ defined by
$$f(\bz;\bw^1,\ldots, \bw^k,\ba,\bb)=\sum_{j=1}^k\alpha_j[\langle\bw^j,\bz\rangle+b_j]_+.$$
Here $\bz \in \mathbb{R}^n$ is the input, $\ba = (\alpha_1, \dots, \alpha_k) \in \mathbb{R}^k$ is a vector of ``coefficients'', $\bw^j=(w^j_1,\ldots, w^j_n)\in \mathbb{R}^n$ is a weight vector associated with the $j$-th unit and $b_j$ is a real parameter (``bias") of the $j$-th unit. 
By simple normalization (e.g., \cite{pan2016expressiveness}), we can assume w.l.o.g. that each $\alpha_j$ is either $+1$ or $-1$.

Networks with rectified linear units (henceforth ReLUs) have gained popularity as they yield state-of-the-art performances in applications such as speech recognition and image classification \cite{krizhevsky2012imagenet,maas2013rectifier}. Several recent works have also explored theoretical aspects of ReLUs \cite{arora2018understanding,goel2016reliably,bach2017breaking,BDL18}.

When training neural networks composed of ReLUs, a popular method is to find, given training data, a set of weights and biases for each gate minimizing the squared loss. More formally,
given a set of $m$ vectors $\bx_1, \ldots ,\bx_m \in \mathbb{R}^n$ along with $m$ real labels $y_1, \ldots, y_m \in \mathbb{R}$, our goal is to find $\bw^1,\ldots \bw^k,\bb$
which minimize the squared training error of the sample:
\begin{equation}\label{equation:relumin}
 \min_{\bw^1,\ldots, \bw^k,\bb}\sum_{i=1}^m (f(\bx_i;\bw^1,\ldots, \bw^k,\ba,\bb)-y_i)^2
\end{equation}
Note that we generally assume that the ``coefficient'' vector $\ba$ is fixed as part of the input to the training problem. (Some of our results apply also when $\ba$ is treated as unknowns. We mention this explicitly when relevant.)

We refer to the optimization problem (\ref{equation:relumin}) as the \emph{ReLU training problem}.
A set of samples $\{(\bx_i, y_i)\}_{i \in [m]}$ is said to be \emph{realizable} if there exist $\bw^1, \cdots, \bw^k, \bb$ which result in zero training error. 
Our goal in this work is to understand the computational complexity of solving the ReLU training problem and study some implications for the problem of (agnostically) \emph{learning} ReLUs.

We are not aware of any hardness results for the ReLU training problem for a single ReLU. For 2 or more ReLU, there are NP-hardness results for networks with different architectures.
In particular, in~\cite{brutzkus2017globally}, the training problem is shown to be hard for a depth-2 \emph{convolutional network} with (at least two) \emph{non overlapping} patches.
Recently,~\cite{BDL18} consider networks similar to us except that the output gate is also a ReLU, instead of a sum gate in our case (see Figure~\ref{subfig:reluofrelu}); they show that, for such networks with three ReLUs, the training problem is NP-hard even for the realizable case.
We remark that our NP-hardness results were obtained independently of those of ~\cite{BDL18} and our NP-hardness proof is different from the proof appearing in~\cite{BDL18}.  The illustrations of our network architectures and the ones considered in~\cite{brutzkus2017globally,BDL18} are presented in Figure~\ref{fig:architecture}.

To the best of our knowledge, these architectural differences render those previous results inapplicable for deriving the hardness results regarding the networks considered in this work.

\begin{figure}
\begin{subfigure}[c]{.20\textwidth}
\centering
\begin{tikzpicture}[scale=0.8]
\node[draw] (relu) at (0, 0) [circle] {\footnotesize ReLU};
\filldraw[black] (-2,-2) circle (2pt) node[below] {$x_1$};
\filldraw[black] (-1,-2) circle (2pt) node[below] {$x_2$};
\filldraw[black] (0,-2) circle (2pt) node[below] {$x_3$};
\node at (1, -2) {$\cdots$};
\filldraw[black] (2,-2) circle (2pt) node[below] {$x_n$};
\draw[->, thick] (relu) -- (0, 2);
\draw[->, thick] (-2, -2) -- (relu);
\draw[->, thick] (-1, -2) -- (relu);
\draw[->, thick] (0, -2) -- (relu);
\draw[->, thick] (2, -2) -- (relu);
\end{tikzpicture}
\subcaption{}
\label{subfig:single}
\end{subfigure}
\hfill
\begin{subfigure}[c]{.20\textwidth}
\centering
\begin{tikzpicture}[scale=0.7]
\node[draw] (pl) at (0, 2) [circle] {\Large +};
\node[draw] (relu1) at (-1.5, 0) [circle] {\footnotesize ReLU};
\node[draw] (relu2) at (1.5, 0) [circle] {\footnotesize ReLU};
\filldraw[black] (-2,-2) circle (2pt) node[below] {$x_1$};
\filldraw[black] (-1,-2) circle (2pt) node[below] {$x_2$};
\filldraw[black] (0,-2) circle (2pt) node[below] {$x_3$};
\node at (1, -2) {$\cdots$};
\filldraw[black] (2,-2) circle (2pt) node[below] {$x_n$};
\draw[->, thick] (relu1) -- (pl);
\draw[->, thick] (relu2) -- (pl);
\draw[->, thick] (pl) -- (0, 4);
\draw[->, thick] (-2, -2) -- (relu1);
\draw[->, thick] (-1, -2) -- (relu1);
\draw[->, thick] (0, -2) -- (relu1);
\draw[->, thick] (2, -2) -- (relu1);
\draw[->, thick] (-2, -2) -- (relu2);
\draw[->, thick] (-1, -2) -- (relu2);
\draw[->, thick] (0, -2) -- (relu2);
\draw[->, thick] (2, -2) -- (relu2);
\end{tikzpicture}
\subcaption{}
\label{subfig:two}
\end{subfigure}
\hfill
\begin{subfigure}[c]{.20\textwidth}
\centering
\begin{tikzpicture}[scale=0.7]
\node[draw] (pl) at (0, 2) [circle] {\footnotesize ReLU};
\node[draw] (relu1) at (-1.5, 0) [circle] {\footnotesize ReLU};
\node[draw] (relu2) at (1.5, 0) [circle] {\footnotesize ReLU};
\filldraw[black] (-2,-2) circle (2pt) node[below] {$x_1$};
\filldraw[black] (-1,-2) circle (2pt) node[below] {$x_2$};
\filldraw[black] (0,-2) circle (2pt) node[below] {$x_3$};
\node at (1, -2) {$\cdots$};
\filldraw[black] (2,-2) circle (2pt) node[below] {$x_n$};
\draw[->, thick] (relu1) -- (pl);
\draw[->, thick] (relu2) -- (pl);
\draw[->, thick] (pl) -- (0, 4);
\draw[->, thick] (-2, -2) -- (relu1);
\draw[->, thick] (-1, -2) -- (relu1);
\draw[->, thick] (0, -2) -- (relu1);
\draw[->, thick] (2, -2) -- (relu1);
\draw[->, thick] (-2, -2) -- (relu2);
\draw[->, thick] (-1, -2) -- (relu2);
\draw[->, thick] (0, -2) -- (relu2);
\draw[->, thick] (2, -2) -- (relu2);
\end{tikzpicture}
\subcaption{}
\label{subfig:reluofrelu}
\end{subfigure}
\hfill
\begin{subfigure}[c]{.25\textwidth}
\centering
\begin{tikzpicture}[scale=0.7]
\node[draw] (pl) at (0, 2) [circle] {\Large +};
\node[draw] (relu1) at (-1.5, 0) [circle] {\footnotesize ReLU};
\node[draw] (relu2) at (1.5, 0) [circle] {\footnotesize ReLU};
\filldraw[black] (-2.5,-2) circle (2pt) node[below] {$x_1$};
\node at (-1.5,-2) {$\cdots$};
\filldraw[black] (-0.5,-2) circle (2pt) node[below] {$x_{\frac{n}{2}}$};
\filldraw[black] (0.5,-2) circle (2pt) node[below] {$x_{\frac{n}{2} + 1}$};
\node at (1.5, -2) {$\cdots$};
\filldraw[black] (2.5,-2) circle (2pt) node[below] {$x_n$};
\draw[->, thick] (relu1) -- (pl);
\draw[->, thick] (relu2) -- (pl);
\draw[->, thick] (pl) -- (0, 4);
\draw[->, thick] (-2.5, -2) -- (relu1);
\draw[->, thick] (-0.5, -2) -- (relu1);
\draw[->, thick] (0.5, -2) -- (relu2);
\draw[->, thick] (2.5, -2) -- (relu2);
\end{tikzpicture}
\subcaption{}
\label{subfig:convnet}
\end{subfigure}
\caption{Diagrams of networks considered in this work and previous works. (\ref{subfig:single}) and (\ref{subfig:two}) are the depth-2 networks we consider, for a single and two ReLUs respectively. For network (\ref{subfig:single}), we show that the training problem is NP-hard (Theorem~\ref{thm:single}) and that even approximating the minimum squared error to within an almost polynomial factor is NP-hard (Theorem~\ref{thm:single-inapprox}). For network (\ref{subfig:two}), we show that the training problem is hard, even in the realizable case (Theorem~\ref{thm:two_rel}). Architectures in (\ref{subfig:reluofrelu}) and (\ref{subfig:convnet}) are considered in~\cite{BDL18} and~\cite{brutzkus2017globally} respectively; the authors show that the training problem for their respective networks is NP-hard even in the realizable case.}
\label{fig:architecture}
\end{figure}
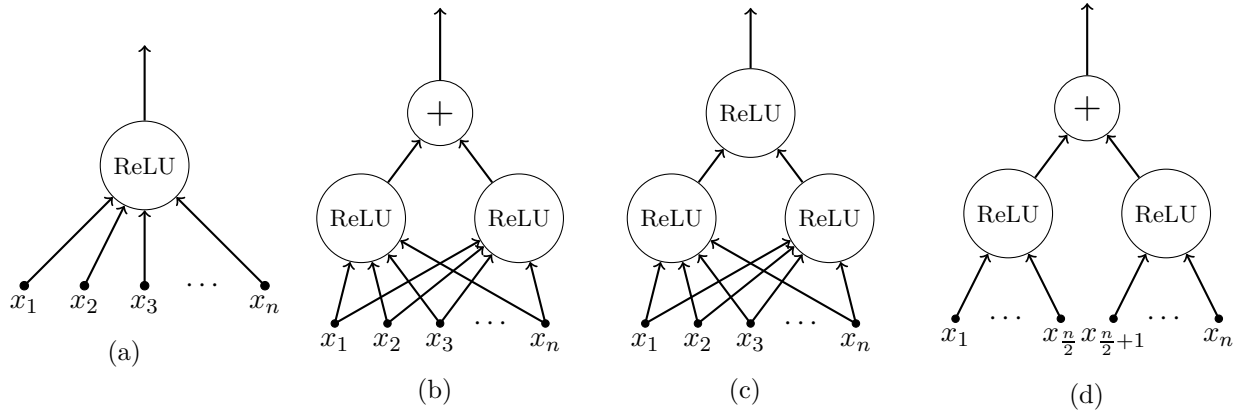

Some works
\cite{arora2018understanding,bach2017breaking} attribute implicitly or explicitly the NP-hardness of the ReLU training problem to \cite{blum1989training} which considers training a neural network with \emph{threshold units}. However, it is not clear (to us) how to derive the NP-hardness of training ReLUs from the hardness results of \cite{blum1989training}.
On the other hand, there are hardness results
with respect to \emph{improper learning} \cite{goel2016reliably,livni2014computational} (i.e., given a set of samples, return an efficiently computable function that is not necessarily a ReLU/network of ReLUs) but these rely on average case assumptions and, hence, do not establish NP-hardness of the ReLU training problem.

On the algorithmic side, Arora et al.~\cite{arora2018understanding} provide a simple and elegant algorithm that exactly solves the ReLU training problem in polynomial time assuming the dimension is an absolute constant; Arora et al.'s algorithm is for the networks we consider, and it has since been also extended to other types of networks~\cite{BDL18}.
Additionally, there have also been works on (agnostic) learning algorithms for ReLUs. Specifically, Goel et al. \cite{goel2016reliably} consider the setting where the inputs to the ReLUs
as well as the weight vectors of the units have norms at most $1$. For this setting, building on kernel methods and tools from approximation theory, they show how to improperly learn a single $n$-variable ReLU up to an additive error of $\epsilon$ in time $2^{O(1/\epsilon)} \cdot poly(n)$. Their result generalizes to depth-2 ReLUs with $k$ units with running time of $2^{O(\sqrt{k}/\epsilon)} \cdot poly(n)$. The algorithm they provide is quite general: it works for arbitrary distribution over input-output pairs, for $\epsilon$
that can be small as $1/\log n$ and also for the reliable setting. They complement their result by showing that even for a single ReLU, when $|\langle \bw, \bx\rangle|$ tends to infinity with $n$, learning $[\langle \bw, \bx\rangle]_+$ (improperly) in time
$g(\epsilon) \cdot poly(n)$ is unlikely as it will result in an efficient algorithm for the problem of learning sparse parities with noise which is believed to be intractable.

\section{Our Results}
We prove both hardness results as well as algorithmic results for training a single ReLU as well as depth-2 ReLUs with $k$ units. 
In terms of hardness, we prove NP-hardness results for the ReLU training problem showing that this problem is hard even for a single ReLU, not only to solve exactly but also to approximate (Section~\ref{sec:singlerelu}).  In Section~\ref{sec:tworelu}, we prove that, in contrast to the case of single ReLU, training 2 ReLUs is NP-hard even in the realizable case\footnote{For completeness we provide a proof in Appendix~\ref{sec:single_rel} that training a single ReLU can be done in polynomial time in the realizable case.}. We remark that this latter result also yields, as an immediate corollary, NP-hardness for training networks considered in~\cite{BDL18}. Our proof is shorter and arguably simpler than the proof appearing in~\cite{BDL18} altough
their result also applies to the case of $k>2$ whereas ours hardness result only applies when $k=2$.


On the algorithmic side, we show, in Section~\ref{sec:learning}, that depth-2 ReLUs can be \emph{properly} (agnostically) learned in time $2^{(k/\epsilon)^{O(1)}}n^{O(1)}$ provided that the inputs and weights of the units belong to the unit ball (see Section~\ref{sec:learning} for precise learning-theoretic definitions). To the best of our knowledge, only \emph{improper} learning algorithms were known before~\cite{goel2016reliably}. 
The insight here is very simple: standard generalization bounds (similar to those used in \cite{goel2016reliably}) imply that it suffices to consider only $(k/\varepsilon)^{O(1)}$ samples. We then observe that the algorithm of~\cite{arora2018understanding} runs in exponential time in the number of samples. Putting these together immediately results in the proper learning algorithm.

We additionally show that, when the coefficients $\alpha_j$'s are all positive, they can be \emph{reliably} properly learned (see Subsection~\ref{sub:reliable} for more details) in similar running time. For the reliable model, we need to also take the advantage of the biases to ensure that there are few false positives. We remark here that Goel et al.~\cite{goel2016reliably} did not allow bias in their ReLUs and hence our algorithm for the reliable model would still be improper for their setting; nevertheless, our output (ReLUs with biases) is still arguably simpler than that of \cite{goel2016reliably} (which is a ``clipped'' of a low degree polynomial). We note that, similar to \cite{goel2016reliably}, our algorithms work also for more general loss functions, as long as they are convex and $O_k(1)$-Lipschitz; we only focus on the squared loss for the simplicity of presentation.


Our lower bounds and algorithms contribute to the quest to understand how neural networks can be trained efficiently
despite NP-hardness results. Specifically, while we prove NP-hardness results for training ReLUs, our learning results (paralleling those of \cite{goel2016reliably} for improper learning) show that efficient training\footnote{A proper learning algorithm immediately yields a polynomial time training algorithm with $\varepsilon$ additive error for any constant $\varepsilon > 0$ (i.e., an additive PTAS).} (up to small additive errors) is possible when weights and inputs of bounded norms are concerned. The exponential dependency of our algorithms on $k/\epsilon$ makes them impractical, and we believe it is of interest to find faster algorithms for properly learning ReLUS.

\section{Hardness of Training a Single ReLU}\label{sec:singlerelu}
We start by showing NP-hardness of training a single ReLU:
\begin{theorem}\label{thm:single}
ReLU training problem for a neural network consisting of a single ReLU is NP-hard.
\end{theorem}

For the simplicity of exposition, we will assume in all our hardness proofs (in this section and Section~\ref{sec:tworelu}) that the biases are equal to zero. In Appendix~\ref{app:bias}, we explain how our proofs can be easily extended to handle non-zero biases.

\begin{proof}
We reduce the set cover problem to the training ReLU problem. Recall that, in the set cover problem, we
are given a set $U=\{1,\ldots, N\}$ along with a family $\mathcal{S}=
\{S_1,\ldots,S_M\}$ of $M$ subsets of $U$. Our goal is
to determine if one can choose $k$ subsets from $\mathcal{S}$ whose union equals $\mathcal{S}$.
Set cover is well known to be NP-hard.


We consider a ReLU with $n = M+2$ variables. For each $S_i \in \mathcal{S}$, we have a variable $w_{S_i}$. We also have two dummy variables $w_1$ and $w_{\epsilon}$.
Let $\alpha_1 = 1$ and $\epsilon=0.01/m^2$.

We introduce the following training points. First, for each $i \in U$, add an $(M + 2)$-dimensional vector having $1$ for the coordinate corresponding to the dummy variable $w_1$, $1$ in all coordinates that correspond to a subset in $S$ containing $i$ and $0$ to all other coordinates. We label this vector by $0$. This labeled data point corresponds to the constraint
\begin{equation}\label{equation:first}
[w_1 + \sum_{i\in S_j} w_{S_j}]_+=0.
\end{equation}
Second, for every $j \in [M]$, add an $(M + 2)$-dimensional vector having $1$ in the $S_j$-th location, $1$ in the coordinate corresponding to $w_{\epsilon}$
and $0$ for all other coordinates. We label it by $\epsilon$. This corresponds to
\begin{equation}\label{equation:second}
[w_{\epsilon} + w_{S_j}]_+=\epsilon.
\end{equation}
We then add a vector having $1$ in the coordinate corresponding to $w_1$ and $0$ elsewhere. We label these vectors by $1$.
This corresponds to
\begin{equation}\label{equation:third}
[w_1]_+=1.
\end{equation}
We also add $(k + 1)$ vectors having $1$ in the coordinate corresponding to $w_{\epsilon}$ and $0$ elsewhere. We label these vectors by $\epsilon$.
These vectors correspond to $(k + 1)$ copies of the constraint
\begin{equation}\label{equation:fourth}
[w_{\epsilon}]_+=\epsilon.
\end{equation}
Finally, we set the target error to be $\epsilon^2k$ where $k$ is the target value in the set cover instance. 
Clearly, this reduction runs in polynomial time.

We now prove the correctness of this reduction.

(YES Case) Assume that there is a set cover of size $k$ consisting of the subsets $S_{j_1}, \ldots, S_{j_k}$ in $\mathcal{S}$. Assigning
$w_{S_{j_1}}=w_{S_{j_2}}=\ldots=w_{S_{j_k}}=-1, w_1=1$, $w_{\epsilon}=\epsilon$ and $0$ to all other variables results in an error of $\epsilon^2\cdot k$. This is because exactly
$k$ of the constraints from (\ref{equation:second}) are violated and each violated constraint contributes $\epsilon^2$ to the squared error. All other constraints are satisfied.

(NO Case) Suppose contrapositively that there is a weight vector $\bw$ that results in an error of at most $\epsilon^2 k$. First, observe that $w_1 \geq 0.9$; otherwise, the squared error from~\eqref{equation:third} is more than $(0.1)^2 \geq \epsilon^2 k$. Observe also that $w_\epsilon \leq 0.2/m$; otherwise, the squared error from~\eqref{equation:fourth} must be more than $(0.2/m - \epsilon)^2 \geq (0.1/m)^2 > \epsilon^2 k$. Moreover, notice that $w_\epsilon$ must be non-negative, since otherwise the $(k + 1)$ copies of~\eqref{equation:fourth} must incur total error of $(k + 1) \epsilon^2 > \epsilon^2 k$.

Our main observation is that the family $\cS_{< -w_\epsilon} = \{S_j: w_{S_j}<-w_\epsilon\}$ is a set cover. The reason is as follows: if there is an element $i \in U$ that is not covered by $\cS_{< -w_\epsilon}$, then $\sum_{i \in S_j} w_{S_j} \geq -w_\epsilon \cdot m \geq -0.2$, which means that the corresponding constraint (\ref{equation:first}) for $i$ will incur already a squared error of at least $(0.7)^2 > \epsilon^2k$ (recall that $k$ is no larger than $m$). Thus, the observation follows.

The last step of the proof is to show that the family $\cS_{<-w_{\varepsilon}}$ contains at most $k$ subsets. To see that this is the case, observe that, for every $S_j \in \cS_{<- w_\varepsilon}$, we have $[w_\epsilon + w_{S_j}]_+ = 0$, meaning that the corresponding constraint~\eqref{equation:second} incurs a squared error of $\epsilon^2$. Since the total squared error is at most $\epsilon^2 k$, we can immediately concludes that at most $k$ subsets belong to $\cS_{< -w_\varepsilon}$.

Thus, $\cS_{< -w_\varepsilon}$ is a set cover with at most $k$ subsets, which completes the NO case of the proof.
\end{proof}

We remark that the above proof (and also that of Theorem~\ref{thm:single-inapprox} below) also works for the case where $\alpha_1$ is treated as an unknown. This is because, if $\alpha_1 = -1$, then the error incurred in~\eqref{equation:third} (resp. in~\eqref{eq:output-wire} below) already exceeds the target error. Thus, it must be that $\alpha_1 = +1$.

\subsection{Hardness of Approximating Minimum Training Error for a Single ReLU}

The reduction above coupled with the fact that set cover is hard to approximate within a factor $O(\log |U|)$~\cite{Feige98} immediately implies that the problem of approximating the minimum training error to within a factor of $O(\log(nm))$ is also hard. In this subsection, we will substantially improve this inapproximability ratio
to an almost polynomial (i.e. $(nm)^{1/\poly \log \log (nm)}$) factor:

\begin{theorem} \label{thm:single-inapprox}
Given an instance of the single ReLU training problem, it is NP-hard to approximate the minimum squared error to within a factor of $(nm)^{1/(\log \log (nm))^{O(1)}}$. 
\end{theorem}


To prove Theorem~\ref{thm:single-inapprox}, we will reduce from the Minimum Monotone Circuit Satisfiability problem, which is formally defined below.

\begin{definition}
A \emph{monotone circuit} is a circuit where each gate is either an OR or an AND gate. We use $|C|$ to denote the number of wires in the circuit.
\end{definition}

\begin{definition}
In the \textsc{Minimum Monotone Circuit Satisfiability$_i$ (MMCS$_i$)} problem, we are given a monotone circuit of depth $i$, and the objective is to assign as few \textsc{True}s as possible to the input wires while ensuring that the circuit is satisfied (i.e. output wire is evaluated to \textsc{True}).

For any monotone circuit $C$, we use $\optmmcs(C)$ to denote the optimum of the MMCS problem on $C$, i.e., the smallest number of input wires need to be set to \textsc{True} so that $C$ is satisfied.
\end{definition}

The hardness of approximating \mmcs\ has long been studied (e.g.~\cite{AlekhnovichBMP01,DinurS04}). By now, this problem is known to be NP-hard to approximate to within a factor of $|C|^{1/(\log \log |C|)^{O(1)}}$:
\begin{theorem}[\cite{DHK15}\footnote{\cite{DHK15} in fact shows that there exists a PCP with $D = (\log \log n)^{O(1)}$ query over alphabet of size $n^{O(1/D)}$ with completeness 1 and soundness $1/n^{\Omega(1)}$. The result we use (Theorem~\ref{thm:mmcs}) follows from their result and from the reduction in Section 3 of~\cite{DinurS04} which shows how to reduce $D$-query PCP over alphabet $F$ with completeness 1 and soundness $s$ to an $\mmcs_3$ instance of size $F^D\poly(n)$ and gap $O(1/s)^{1/D}/D$. Plugging this in immediately implies the hardness we use.}] \label{thm:mmcs}
$\mmcs_3$ is NP-hard to approximate to within $|C|^{1/(\log \log |C|)^{O(1)}}$ factor.
\end{theorem}

The main result of this subsection is that, for any $\ell > 0$, there is a polynomial-time reduction from $\mmcs_\ell$ to the problem of minimizing the training error in single ReLU such that the optimum of the latter is proportional to the optimum of the former. From Theorem~\ref{thm:mmcs} above, this immediately implies Theorem~\ref{thm:single-inapprox}. The reduction is stated and proved below.


\begin{theorem}\label{thm:montonehard}
For every $\ell > 0$, there is a polynomial-time reduction that on a depth-$\ell$ monotone circuit $C$ produces samples $\{(\bx_i, y_i)\}_{i \in [m]}$ such that the minimum squared training error for these samples among single ReLUs is $\optmmcs(C) / (10|C|)^{2\ell + 2}$.
\end{theorem}

\begin{proof}
Let $\varepsilon := 1/(10|C|)^{\ell + 1}$. We consider a ReLU with $n = |C| + 1$ variables. For each wire $j$, we create a variable $w_j$. Additionally, we have a dummy variable $w_\epsilon$. (In the desired solution, we want $w_j$ to be 1 iff the wire is evaluated to \textsc{True} and 0 otherwise, and $w_\epsilon = \epsilon$.) Lastly, we set $\alpha_1 = 1$.

{\bf Dummy Variable Constraint.} We add the following constraint
\begin{align} \label{eq:dummy-eps}
[w_\epsilon]_+ = \epsilon.
\end{align}

{\bf Input Wire Constraint.} For each input wire $i$, we add the constraint
\begin{align} \label{eq:input-wire}
[w_\epsilon - w_i]_+ = \epsilon.
\end{align}

{\bf Output Wire Constraint.} For the output wire $o$, we add the constraint
\begin{align} \label{eq:output-wire}
[w_o]_+ = 1.
\end{align}

{\bf OR Gate Constraint.} For each OR gate with input wires $i_1, \dots, i_k$ and output wire $j$, we add the constraint
\begin{align} \label{eq:or-gate}
[w_j - w_{i_1} - \cdots - w_{i_k}]_+ = 0.
\end{align}

{\bf AND Gate Constraint.} For each AND gate with input wires $i_1, \dots, i_k$ and output wire $j$, we add the following $k$ constraints:
\begin{align} \label{eq:and-gate}
[w_j - w_{i_1}]_+ = 0, \cdots, [w_j - w_{i_k}]_+ = 0.
\end{align}

We will now show that the minimum squared training error is exactly $\optmmcs(C) \cdot \varepsilon^2$. First, we will show that the error is at most $\optmmcs(C) \cdot \varepsilon^2$. Suppose that $\phi$ is an assignment to $C$ with $\optmmcs(C)$ \textsc{True}s that satisfies the circuit. We assign $w_\epsilon = \epsilon$, and, for each wire $j$, we assign $w_j$ to be 1 if the wire $j$ is evaluated to be \textsc{True} on input $\phi$ and 0 otherwise. It is clear that every constraint is satisfied except the input wire constraints~\eqref{eq:input-wire} for the wires that are assigned to \textsc{True} by $\phi$. There are exactly $\optmmcs(C)$ such wires, and each contributes $\varepsilon^2$ to the error; as a result, the training error of such weights is exactly $\optmmcs(C) \cdot \varepsilon^2$.

Next, we will show that the minimum squared training error is at least $\optmmcs(C) \cdot \varepsilon^2$. Suppose for the sake of contradiction that the minimum error $\delta$ is less than $\optmmcs(C) \cdot \varepsilon^2$. Observe that, from $\optmmcs(C) \leq |C|$ and from our choice of $\varepsilon$, we have
\begin{align} \label{eq:delta-bound}
\delta < |C| \cdot \varepsilon^2 < 0.1
\end{align}

Consider an assignment $\phi$ that assigns each input wire $i$ to be \textsc{True} iff $w_i \geq w_\epsilon$. The following proposition bounds the weight of every \textsc{False} wire.

\begin{proposition} \label{prop:height-bound}
For any wire $j$ at height $h$ that is evaluated to \textsc{False} on $\phi$, $w_j \leq (2|C|)^h \cdot (\varepsilon + \sqrt{\delta})$.
\end{proposition}

Note that we define the height recursively by first letting the heights of all input wires be zero and then let the height of the output wire of each gate $G$ be one plus the maximum of the heights among all input wires of $G$. The proof of this proposition, which is based on a simple induction, is deferred to Appendix~\ref{app:induction}.

Now, consider the output wire $o$. We claim that $o$ must be evaluated to \textsc{True} on $\phi$. Otherwise, Proposition~\ref{prop:height-bound} ensures that $w_o$ is at most
\begin{align*}
(2|C|)^{\ell} \cdot (\varepsilon + \sqrt{\delta}) \stackrel{\eqref{eq:delta-bound}}{<} (2|C|)^\ell \cdot (\varepsilon + \sqrt{|C|} \cdot \varepsilon) \leq 0.1,
\end{align*}
where the second inequality comes from our choice of $\varepsilon$. This would mean that the squared error incurred in~\eqref{eq:output-wire} is at least $0.81 > \delta$. Thus, it must be that $\phi$ satisfies $C$.

Moreover, since $\phi$ assigns each input wire $i$ to be \textsc{True} iff $w_i \geq w_\epsilon$, each input wire that is assigned \textsc{True} incurs a squared error of $\varepsilon^2$ from~\eqref{eq:input-wire}. Thus, the number of input wires assigned \textsc{True} is at most $\frac{\delta}{\varepsilon^2} < \optmmcs(C)$, which is a contradiction as we argued that $\phi$ satisfies $C$. 
\end{proof}


Observe that, in both Theorem~\ref{thm:montonehard} and Theorem~\ref{thm:single}, the target squared error tends to zero as the dimension tends to infinity. However, this is not an issue: if the norms of the sample vectors are not required to be bounded, then we can simply multiply them by any factor to make the error arbitrarily large. On the other hand, our learning algorithm below implies that, when the norms of samples and weights of ReLUs are bounded, we can approximate the minimum training error for $k$ ReLUs up to an additive error of $\epsilon$ in time $2^{(k/\varepsilon)^{O(1)}}\cdot poly(n)$.

\section{NP-hardness of Training Two ReLUs}\label{sec:tworelu}

We next prove that, for two ReLUs, not only the training problem is NP-hard, but it is NP-hard to even determine whether the samples are realizable. (We remark that this also rules out any multiplicative approximation for the training problem with two ReLUs.) This is in contrast with the single ReLU case, where the realizable case is easy to solve (see Appendix~\ref{sec:single_rel}).

\begin{theorem}\label{thm:two_rel}
It is NP-hard to determine, given labeled samples of a network consisting of two ReLUs, whether it is possible to assign weights to the units such that the training error is $0$.
\end{theorem}

\begin{proof}
We reduce from the 3SAT problem. Recall that, in the 3SAT problem, we are given 3CNF formulas with $M$ clauses on $N$ Boolean variables $X_1, \dots, X_N$ and we would like to determine whether there exists an assignment that satisfies the formula. 

The reduction proceeds as follows. Let $n = N + 1$ and $m = 2N + M + 1$. We view the $n$-th coordinate of each sample as a coefficient of dummy variables which we will refer to as $v^1$ ($= w^1_n$) and $v^2$ ($= w^2_n$). Moreover, let $\alpha_1 = \alpha_2 = 1$.

The first sample has only one non-zero coordinate corresponding to $v$ which is set to one and has label $4$, i.e., this corresponds to
\begin{align} \label{eq:dummy}
[v^1]_+ + [v^2]_+ = 4.
\end{align}

Next, for every variable $X_i$, we add constraints
\begin{align}
[w^1_i]_+ + [w^2_i]_+ &= 1, \label{eq:boolean1} \\
[-w^1_i]_+ + [-w^2_i]_+ &= 1. \label{eq:boolean2}
\end{align}
Finally, for each clause $C_j = (b_1 \vee b_2 \vee b_3)$, we add a constraint as follows. For $p = 1, 2, 3$, let $X_{i_p}$ denote the variable corresponding to the literal $b_p$; moreover, let $n_p$ be +1 if the literal is positive and -1 otherwise. We then add the following constraint for this clause:
\begin{align} \label{eq:clause}
[-v^1 - n_1 \cdot w^{1}_{i_1} - n_2 \cdot w^{1}_{i_2} - n_3 \cdot w^{1}_{i_3}]_+ + \nonumber \\
[-v^2 - n_1 \cdot w^{2}_{i_1} - n_2 \cdot w^{2}_{i_2} - n_3 \cdot w^{2}_{i_3}]_+ = 0.
\end{align}

The reduction clearly runs in polynomial time. Next, we argue the correctness of the reduction. 

(YES Case) We will start with the YES case. Suppose that the formula is satisfiable. That is, there exists an assignment $\phi: [N] \to \{0, 1\}$ that satisfies all clauses. Set $v^1 = 1, v^2 = 3$ and, for every $i \in [N]$, $w^1_i = 2\phi(i) - 1$ and $w^2_i = 1 - 2\phi(i)$. It is easy to verify that all constraints are satisfied, i.e., that the samples are realizable by a sum of two ReLUs with boolean weights.

(NO Case) We will prove the contrapositive. Suppose that there exist $\bw^1, \bw^2 \in \mathbb{R}^n$ that satisfies all the constraints. Constraint~\eqref{eq:dummy} implies that at least one of $v^1$ and $v^2$ must be at most 2; we assume w.l.o.g. that $v^1 \leq 2$. It is then easy to see that any $w^1_i, w^2_i$ that satisfy~\eqref{eq:boolean1} and \eqref{eq:boolean2} must satisfy $w^1_i, w^2_i \in \{\pm 1\}$. Define an assignment $\phi$ for the 3CNF formula by $\phi(i) = 1$ if $w^1_i = +1$ and $\phi(i) = 0$ if $w^1_i = -1$. Finally, observe that,~\eqref{eq:clause} implies that, for every clauses $C_j$, at least one of $n_1 \cdot w^1_{i_1}$, $n_2 \cdot w^1_{i_2}$ and $n_3 \cdot w^1_{i_3}$ must be $+1$ (as, otherwise, the sum $-v^1 - n_1 \cdot w^{1}_{i_1} - n_2 \cdot w^{1}_{i_2} - n_3 \cdot w^{1}_{i_3}$ must be at least $-2 + 1 + 1 + 1 > 0$); this indeed means that the corresponding literal must be set to true by $\phi$. As a result, $\phi$ must satisfy all the clauses in the formula, as desired.
\end{proof}




We remark that, once again, the hardness in Theorem~\ref{thm:two_rel} applies even to the case where $\alpha_1, \alpha_2$ are treated as unknowns. Specifically, \eqref{eq:boolean1} and \eqref{eq:boolean2} already enforce both $\alpha_1$ and $\alpha_2$ to be positive.

\section{Learning ReLUs} \label{sec:learning}


We follow the agnostic learning model for real-valued functions from~\cite{Haussler92,KearnsSS94}. (This is in turn based on the PAC learning model for $\{0, 1\}$-valued functions~\cite{Valiant84}.) A \emph{concept class} $\cC: \cY^{\cX}$ is any set of functions from $\cX$ to $\cY$. We say that a concept class $\cC$ is \emph{agnostically learnable} with respect to a loss function $\ell: \cY^2 \to \mathbb{R}$ if, for every $\delta, \varepsilon > 0$, there is an algorithm $\cA$ such that, for any distribution $\cD$ over $\cX \times \cY$, receives as input independent random samples from $\cD$ and outputs a hypothesis $h \in \cY^{\cX}$ such that, with probability $1 - \delta$, $$\cL(h; \cD) \leq \inf_{c \in \cC} \cL(c; \cD) + \varepsilon$$ where $\cL(f; \cD) := \E_{(\bx, y) \sim \cD}[\ell(f(\bx), y)]$ is the expected loss for $f$ over $\cD$.
If the output hypothesis $h$ belongs to the concept class $\cC$, then we said that it is \emph{properly} agnostically learnable.

Another model we consider is the \emph{reliable} agnostic learning model; in the real-valued setting, this model was first defined in~\cite{goel2016reliably}, based on the model of~\cite{KalaiKM12} for the standard PAC learning model. Informally speaking, reliability puts more emphasis on false positives, i.e., $(\bx, y)$ supported on $\cD$ such that $y = 0$ but $h(\bx) > 0$. The additional requirement is that such false positives should only happen with probability $\leq \varepsilon$. (For motivations of the model, see e.g.~\cite{goel2016reliably}.)

More formally, we say that a concept class $\cC$ is \emph{reliably agnostically learnable} with respect to loss function $\ell$ if, for every $\delta, \varepsilon > 0$, there is an algorithm $\cA$ such that, for any distribution $\cD$ over $\cX \times \cY$, takes independent random samples from $\cD$ and outputs a hypothesis $h$ such that, with probability $1 - \delta$, the following holds:
\begin{align*}
\cL_{=0}(h; \cD) &\leq \varepsilon, \\
\cL(h; \cD) &\leq \inf_{c \in \cC^+(\cD)} \cL(c; \cD) + \varepsilon.
\end{align*}
where $\cL_{=0}(h; \cD) = \Pr_{(\bx, y) \sim \cD}[h(\bx) > 0 \wedge y = 0]$ is the probability of false positive, and $\cC^+(\cD) = \{c \in \cC \mid \cL_{=0}(c; \cD) = 0\}$ denote all functions in the concept class that (with probability 1) do not admit any false positives. Similar to before, we say that $\cA$
is \emph{proper} if $h \in \cC$.

Before we move on, we remark that, in the reliable model, the error $\cL(h; \cD)$ is only compared to $\cL(c; \cD)$ for $c$ that does not admit any false positives, unlike in the (non-reliable) agnostic learning model where all $c \in \cC$ are considered. In other words, the fact that a concept class $\cC$ is reliably agnostically learnable does not necessarily imply that it is agnostically learnable. It is also not hard to verify that the fact that a concept class $\cC$ is agnostically learnable does not imply that it is reliably agnostically learnable.

\subsection{Our Results}

We now proceed to state our results. The concept classes we consider are the classes of sums of $k$ ReLUs, where each weight vector has norm at most one, and the distribution $\cD$ is allowed to be any distribution on the unit ball. More specifically, the class ReLU$(n, k)$, which represent the sums of $k$ ReLUs, is defined as follows:

\begin{definition}
For any $n, k \in \mathbb{N}$, $\bw^1, \dots, \bw^k \in \cB^n$ and $b_1, \dots, b_k \in [-1, 1]$, let $\relu_{\bw^1, \dots, \bw^k}^{b_1, \dots, b_k}: \cB^n \to [0, 2k]$ denote the function $\bx \mapsto \sum_{j=1}^k [\left<\bw_j, \bx\right> + b_j]_+$.

Let ReLU$(n, k)$ denote the class $\{\relu_{\bw^1, \dots, \bw^k}^{b_1, \dots, b_k} \mid \bw^1, \dots, \bw^k \in \cB^n, b_1, \dots, b_k \in [-1, 1]\}$.
\end{definition}

We show that, for any fixed number of ReLUs $k$ and error parameter $\varepsilon > 0$, the class above can be efficiently agnostically properly learned (both reliably and non-reliably), as stated below.

\begin{theorem} \label{thm:learning}
For any $n, k \in \mathbb{N}$, ReLU$(n, k)$ can be agnostically properly learned for the squared loss function in time $2^{O(k^5/\varepsilon^2)} \cdot (n/\delta)^{O(1)}$ time.
\end{theorem}

\begin{theorem} \label{thm:learning-reliable}
For any $n, k \in \mathbb{N}$, ReLU$(n, k)$ can be agnostically reliably properly learned for the squared loss function in time $2^{O(k^7/\varepsilon^4)} \cdot (n/\delta)^{O(1)}$ time.
\end{theorem}

Observe that both Theorems consider learning the sum of $k$ ReLUs, i.e., when $\alpha_1 = \cdots = \alpha_k = 1$. For Theorem~\ref{thm:learning}, the same result holds for arbitrary coefficients (with a similar proof). This theorem can be further generalized to the case where the coefficients $\alpha_1, \dots, \alpha_k$ are unknowns with only $2^k$ multiplicative overhead to the running time, by enumerating all $\alpha_1, \dots, \alpha_k \in \{\pm 1\}$. On the other hand, it is unclear how to extend the algorithm in Theorem~\ref{thm:learning-reliable} to work for negative coefficients; however, we note that it is not even clear whether ``reliable'' makes sense in this case, since the predicted values can take negative values.

Our results above should be compared to those of~\cite{goel2016reliably} who showed similar results, except that their algorithm is \emph{improper}: their output is a (``clipped'' of) low-degree polynomial, as opposed to sums of ReLUs (which our algorithm outputs).  While our algorithm is advantageous to theirs in this sense, theirs is faster\footnote{We do not attempt to optimize our running time, for the sake of simplicity. Nevertheless, it is clear that our approach cannot go beyond $2^{O(k^2/\varepsilon^2)} \cdot (n/\delta)^{O(1)}$ time, which is still slower than the algorithms of~\cite{goel2016reliably}.} and extends to a larger class of networks.

Our proof is simple. It first applies \emph{generalization bounds} (similar to \cite{goel2016reliably}) which implies that it suffices to take $(k/\varepsilon)^{O(1)}$ samples and solve (even approximately) the training problem on these samples.
Hence, by invoking the algorithm from Arora et al.'s work~\cite{arora2018understanding} (see Lemma~\ref{lem:exp-algo}), we immediately get Theorem~\ref{thm:learning}.

To ensure the reliability guarantee (Theorem~\ref{thm:learning-reliable}), we do not immediately output the minimizer $h$ from Arora et al.'s algorithm. Rather, we ``shift'' the biases by subtracting them with a small number. By doing so, for any $\bx$ such that $y = 0$ and $h(\bx)$ is non-zero but not too large, the modified hypothesis makes sure that $(\bx, y)$ is not a false positive (see~\eqref{eq:to-cont} below). This is a difference between our proof and the one used in~\cite{goel2016reliably} where all biases are assumed to be zero and hence they need to ``clip'' their hypothesis instead. This is also where we need the positivity of $\alpha_j$'s; if $\alpha_j$'s are allowed to be negative, it could be that $h(\bx)$ is small but it remains non-zero after bias shifts.

\subsection{Generalization Bounds}

Before we get to our proofs, we state the necessary generalization bounds; these are exactly the same as those used in~\cite{goel2016reliably}. (See Section 2.5 there.)

\begin{theorem}[\cite{BM02}] \label{thm:gen-err}
Let $\cD$ be a distribution over $\cX \times \cY$ and let $\ell: \cY \times \cY \to \mathbb{R}$ be a $b$-bounded loss function that is $L$-Lispschitz in its first argument. Let $\cF \subseteq (\cY')^{\cX}$ and for any $f \in \cF$, let $\cL(f; \cD) := \E_{(\bx, y) \sim \cD}[\ell(f(\bx), y)]$ and $\hcL(f; S) := \frac{1}{m} \sum_{i=1}^m \ell(f(\bx_i), y_i)$, where each sample $(\bx_i, y_i) \in S$ is drawn independently uniformly at random according to $\cD$. Then, for any $\delta > 0$, with probability at least $1 - \delta$, the following is true for all $f \in \cF$:
\begin{align*}
|\cL(f; \cD) - \hcL(f; S)| \leq 4L \cdot \cR_m(\cF) + 2b \sqrt{\frac{\log(1/\delta)}{m}}
\end{align*}
where $\cR_m(\cF)$ is the Rademacher complexity of $\cF$.
\end{theorem}

\begin{theorem}[\cite{KST08}] \label{thm:KST}
Let $\cX \subseteq \cB^n$ and $\cW = \{\bx \mapsto \left<\bx, \bw\right> \mid \|w\|_2 \leq 1\}$. Then, $\cR_m(\cW) \leq \sqrt{\frac{1}{m}}$.
\end{theorem}

\begin{fact} \label{fact:subadd}
Let $\cF_1, \cF_2 \subseteq \mathbb{R}^{\cX}$ and $\cF = \{f_1 + f_2 \mid f_1 \in \cF_1, f_2 \in \cF_2\}$. Then, $\cR_m(\cF) \leq \cR_m(\cF_1) + \cR_m(\cF_2)$.
\end{fact}

\begin{theorem}[\cite{BM02,LT91}] \label{thm:lips}
Suppose that $\psi: \mathbb{R} \to \mathbb{R}$ is $L_{\psi}$-Lipschitz  and $\psi(0) = 0$. Let $\cY \subseteq \cR$.
For any $\cF \subseteq \cY^{\cX}$, it holds that $\cR_m(\{\psi \circ f \mid f \in \cF\}) \leq 2 \cdot L_\psi \cdot \cR_m(\cF)$.
\end{theorem}

\subsection{Arora et al.'s Training Algorithm} \label{sec:exp-algo}

Another ingredient is the algorithm of~\cite{arora2018understanding}, which runs in time $m^{O(kn)}$ and output the optimal training error (to within arbitrarily small accuracy). We observe that, for $m \ll n$, the running time becomes $2^{km} \cdot poly(n, m, k)$ which is even faster:

\begin{lemma} \label{lem:exp-algo} 
There is an $2^{km} \cdot poly(n, m,1/\beta,C)$-time algorithm that, given samples $\{(\bx_i, y_i)\}_{i \in [m]}$ where $\bx_i \in \mathbb{R}^n$ and an accuracy parameter $\beta \in (0,1)$, finds $\bw_1, \dots, \bw_k \in \cB^n$ and $\bb \in [-1, 1]^k$ that minimizes the squared training error
up to an additive error of $\beta$, where $C$ is the bit complexity of the numbers in the input.
Furthermore, there is an algorithm with the same running time that finds $\bw_1, \dots, \bw_k \in \cB^n, \bb \in [-1, 1]^k$ that minimizes the squared training error to within $\beta$ additive error subjects to additional constraints that $\sum_{j \in [k]} [\left<\bw_j, \bx_j\right> + b_j]_+ = 0$ for all $i \in [m]$ such that $y_i = 0$.
\end{lemma}

Since the result stated here is slightly different than the version in~\cite{arora2018understanding}, we sketch its proof in Appendix~\ref{sec:exp-algo}.


\subsection{Properly Learning ReLUs}

We now proceed to prove Theorem~\ref{thm:learning}. When we invoke the algorithm from Lemma~\ref{lem:exp-algo}, we will ignore the accuracy parameter $\beta$ and pretend that the algorithm output an actual optimal solution. This is with out loss of generality as in the applications below we can always set $\beta$ sufficiently small such that it becomes negligible. We only choose to ignore it because the proof is much cleaner this way.

\begin{proof}[Proof of Theorem~\ref{thm:learning}]
First, let us describe the algorithm. Given samples $S = \{(\bx_i, y_i)\}_{i \in [m]}$ where\footnote{If there are more than $m$ samples, just consider $m$ of them.}
\begin{align*}
m = \left\lceil \frac{10^{10} \cdot k^4 \cdot (1 + \log(1/\delta))}{\varepsilon^2} \right\rceil,
\end{align*}
we use the algorithm in Lemma~\ref{lem:exp-algo} to solve for $\bw_1, \dots, \bw_k, b_1, \dots, b_k$ that minimizes the training error. Then, output the hypothesis $h = \relu_{\bw_1, \dots, \bw_k}^{b_1, \dots, b_k}$.

Clearly, the algorithm is a proper learning algorithm (i.e. $h \in$ ReLU$_{\alpha_1, \dots, \alpha_k}(n, k)$). Furthermore, it runs in time $2^{km} poly(n, m) = 2^{O(k^5/\varepsilon^2)} poly(n, 1/\delta)$.

Thus, we are left to bound the error $\cL(h; \cD)$. Observe that, from Theorems~\ref{thm:KST} and~\ref{thm:lips}, we have $\cR_m(\text{ReLU}(n, 1)) \leq \frac{2}{\sqrt{m}}$. Hence, from Fact~\ref{fact:subadd}, we have $\cR_m(\text{ReLU}(n, k)) \leq \frac{2k}{\sqrt{m}}$. Since the squared loss function is $(4k)$-Lipschitz and $(4k^2)$-bounded in $[0, 2k]^2$, Theorem~\ref{thm:gen-err} implies that the following holds for all $f \in \text{ReLU}(n, k)$ with probability at least $1 - \delta$:
\begin{align} \label{eq:gen-err-simple}
|\cL(f; \cD) - \hcL(f; S)|
\leq \frac{\varepsilon}{2}.
\end{align}

For any $c \in \text{ReLU}(n, k)$, since $h$ minimizes the training error,
\begin{align} \label{eq:opt-simple}
\hcL(h; S) \leq \hcL(c; S).
\end{align}

As a result, we have
\begin{align*}
\cL(h; S) \stackrel{\eqref{eq:gen-err-simple}}{\leq} \hcL(h; S) + \frac{\varepsilon}{2} \stackrel{\eqref{eq:opt-simple}}{\leq} \hcL(c; S) + \frac{\varepsilon}{2} \stackrel{\eqref{eq:gen-err-simple}}{\leq} \cL(c; S) + \varepsilon 
\end{align*}
which concludes the proof.
\end{proof}

\subsection{Properly Reliably Learning ReLUs}\label{sub:reliable}

\begin{proof}[Proof of Theorem~\ref{thm:learning-reliable}]
Again, we start with our algorithm. Given samples $S = \{(\bx_i, y_i)\}_{i \in [m]}$ where
\begin{align*}
m = \left\lceil \frac{10^{10} \cdot k^6 \cdot \log(2/\delta)}{\varepsilon^4} \right\rceil.
\end{align*}
We use the algorithm from Lemma~\ref{lem:exp-algo} to solve for $\bw_1, \dots, \bw_k, b_1, \dots, b_k$ that minimizes the training error for the $m$ samples subject to the additional constraints that, for every sample $\bx_i$ with $y_i = 0$, we have $\sum_{j=1}^k [\left<\bw_j, \bx_i\right> + b_j]_+ = 0$. Then, let $b'_j= \max\{-1, b_j - \gamma\}$ for all $j = 1, \dots, k$ where $\gamma = \frac{\varepsilon}{12k^2}$ and output the hypothesis $h_{\text{shifted}} = \relu_{\bw_1, \dots, \bw_k}^{b'_1, \dots, b'_k}$.

This is clearly a proper learning algorithm and runs in $2^{km} poly(nm) = 2^{O(k^7/\varepsilon^4)} poly(n/\delta)$ time.

Thus, we are left to bound the loss. To do so, first recall (from the proof of Theorem~\ref{thm:learning}) that $\cR_m(\text{ReLU}(n, k)) \leq \frac{2k}{\sqrt{m}}$. Recall also that, for reliable learning, we need to bound two losses:
\begin{align*}
\cL_{=0}(h_{\text{shifted}}; \cD) &= \Pr_{(\bx, y) \sim \cD}[h_{\text{shifted}}(\bx) \ne 0 \wedge y = 0] \\
\cL(h_{\text{shifted}}; \cD) &= \E_{(\bx, y) \sim \cD}\left[(h_{\text{shifted}}(\bx) - y)^2\right].
\end{align*}

For convenience, let $h =  \relu_{\bw_1, \dots, \bw_k}^{b_1, \dots, b_k}$ be the minimizer before bias shifts.

\paragraph{Bounding $\cL_{=0}$.}  Define another loss function $\ell_{\gamma\text{-cont}}$ where $\ell_{\gamma\text{-cont}}(y', y) = 0$ for all $y \ne 0$ and
\begin{align*}
\ell_{\gamma\text{-cont}}(y', 0)
&=
\begin{cases}
0 & \text{ if } y' \leq 0 \\
y'/\gamma & \text{ if } y' \in (0, \gamma) \\
1 & \text{ if } y' \geq \gamma
\end{cases}
\end{align*}
Since $\ell_{\gamma\text{-cont}}$ is $(1/\gamma)$-Lipschitz and $1$-bounded on $[0, 2k]^2$, Theorem~\ref{thm:gen-err} implies that the following holds for all $f \in \text{ReLU}(n, k)$ with probability at least $1 - \delta/2$:
\begin{align} \label{eq:gen-err-cont}
|\cL_{\gamma\text{-cont}}(f; \cD) - \hcL_{\gamma\text{-cont}}(f; S)|
\leq \varepsilon.
\end{align}

Observe that, if $h(\bx) \leq \gamma$, then the bias shifts ensure that $h_{\text{shifted}}(\bx) = 0$; this is because $h(\bx) \leq \gamma$ implies that $\left<\bw_j, \bx\right> + b_j \leq \gamma$ for all $j \in [k]$, which means that $\left<\bw_j, \bx\right> + b'_j \leq 0$. (Note that this is the place where we need positivity of $\alpha_j$'s.)
As a result, we have
\begin{align} \label{eq:to-cont}
\cL_{=0}(h_{\text{shifted}}; \cD) \leq \cL_{\gamma\text{-cont}}(h; \cD).
\end{align}

Combining~\eqref{eq:to-cont} and~\eqref{eq:gen-err-cont}, we can conclude that the following holds with probability $1 - \delta/2$:
\begin{align} \label{eq:realizable-bound}
\cL_{=0}(h_{\text{shifted}}; \cD) \stackrel{\eqref{eq:to-cont}}{\leq} \cL_{\gamma\text{-cont}}(h; \cD) \stackrel{\eqref{eq:gen-err-cont}}{\leq} \varepsilon, 
\end{align}
where the last inequality also comes from the fact that $h(\bx_i) = 0$ for all $i$ with $y_i = 0$, i.e., $\hcL_{\gamma\text{-cont}}(h; S) = 0$.

\paragraph{Bounding $\cL$.} Notice that, for any $\bx \in \cB^n$, $|h(\bx) - h_{\text{shifted}}(\bx)| \leq k \cdot \gamma$. Since the squared loss function is $(4k)$-Lipschitz in the domain $[0, 2k]^2$, it holds that
\begin{align} \label{eq:shift-error}
|\cL(h; \cD) - \cL(h_{\text{shifted}}; \cD)| \leq (4k) \cdot k \cdot \gamma = \varepsilon / 3.
\end{align}

Since $\ell$ is $(4k)$-Lipschitz and $(4k^2)$-bounded on $[0, 2k]^2$, Theorem~\ref{thm:gen-err} implies that, with probability $1 - \delta/2$, the following holds for all $f \in \text{ReLU}(n, k)$:
\begin{align} \label{eq:gen-err}
|\cL(f; \cD) - \hcL(f; S)|
\leq \frac{\varepsilon}{3},
\end{align}

Finally, let $c$ be any function in ReLU$(n, k)$ such that $c(\bx) = 0$ for all $(\bx, y)$ in the support of $\cD$ such that $y = 0$. From how $h$ is computed, we must have
\begin{align} \label{eq:opt}
\hcL(h; S) \leq \hcL(c; S)
\end{align}

By combining the above bounds, the following holds with probability $1 - \delta/2$:
\begin{align*}
\cL(h_{\text{shifted}}; \cD)
&\stackrel{\eqref{eq:shift-error}}{\leq} \varepsilon/3 + \cL(h; \cD) \\
&\stackrel{\eqref{eq:gen-err}}{\leq} 2\varepsilon/3 + \hcL(h; S) \\
&\stackrel{\eqref{eq:opt}}{\leq} 2\varepsilon/3 + \hcL(c; S) \\
&\stackrel{\eqref{eq:gen-err}}{\leq} \varepsilon + \cL(c; \cD),
\end{align*}
which, together with~\eqref{eq:realizable-bound}, completes our proof.
\end{proof}

\section*{Acknowledgments}
We are indebted to Adam Klivans for useful comments on an preliminary version of this work and for his suggestion to study the bounded norm case.
We thank Amir Globerson and Amit Daniely for helpful discussions. We thank an anonymous reviewer for pointing out an error in a previous version of this work. 
\bibliographystyle{alpha}
\bibliography{relu}

\appendix

\section{Dealing with Biases In NP-hardness Proofs}
\label{app:bias}

As stated earlier, the proofs for NP-hardness results in the main body of the paper assumes that the biases $b_1, \cdots, b_k$ are all zeros. However, all NP-hardness results apply even for unknown $b_1, \dots, b_k$, with little to no change. We elaborate on this below.

\subsection{NP-hardness of Training a Single ReLU}

For Theorem~\ref{thm:single}, the same reduction establishes NP-hardness result when there is a bias variable $b_1$ in the ReLU. In the YES case, we can simply set $b_1$ to $0$. In the NO case, we can get an assignment with the same squared error and no bias by replacing $w_{\epsilon}$ by $w_{\epsilon}-b_1$ and $w_1$ by $w_1-b_1$, and thereafter use the same arguments as in the proof of Theorem~\ref{thm:single}.

\subsection{NP-hardness of Training Two ReLUs}

For Theorem~\ref{thm:two_rel}, we need to add two dummy variables $v^1_{\text{dummy}}$ and $v^2_{\text{dummy}}$, and add the following constraints
\begin{align*}
[b_1]_+ + [b_2]_+ &= 0 \\
[v^1_{\text{dummy}} + b_1]_+ + [v^2_{\text{dummy}} + b_2]_+ &= 1, \\ [-v^1_{\text{dummy}} + b_1]_+ + [-v^2_{\text{dummy}} + b_2]_+ &= 1, \\ [2v^1_{\text{dummy}} + b_1]_+ + [2v^2_{\text{dummy}} + b_2]_+ &= 2, \\ [-2v^1_{\text{dummy}} + b_1]_+ + [-2v^2_{\text{dummy}} + b_2]_+ &= 2.
\end{align*}
The YES case proceeds the same as before, by additionally setting $v^1_{\text{dummy}} = 1, v^2_{\text{dummy}} = -1$ and $b_1 = b_2 = 0$. In the NO case, these constraints force $b_1$ and $b_2$ to both be zero. The rest of the proof remains unchanged.

\subsection{NP-hardness of Approximating Training Error of a Single ReLU}

For Theorem~\ref{thm:single-inapprox}, we add a dummy variable $v_{\text{dummy}}$, and add the following constraints:
\begin{align*}
[b_1]_+ &= 0, \\
[v_{\text{dummy}} + b_1]_+ &= 1,\\
[2v_{\text{dummy}} + b_1]_+ &= 2.
\end{align*}
Again, it is simple to see that the minimum training error is at most is at most $\opt_{\mmcs}(C) \cdot \varepsilon^2$, by additionally setting $v_{\text{dummy}} = 1$ and $b_1 = 0$.

The other direction of the proof (i.e. that the minimum training error is at least $\opt_{\mmcs}(C) \cdot \varepsilon^2$) is more delicate. First, one needs to observe that $|b_1|$ cannot be more than $2\sqrt{\delta}$; otherwise, one of the three additional constraints must contribute to more than $\delta$ to the training error. Then, we can once again use induction as before to prove a statement similar to Proposition~\ref{prop:height-bound}, except that the bound will now be $(2|C|)^h \cdot (\varepsilon + 3\sqrt{\delta})$. The rest of the proof proceeds as before. Once again, we will be able to conclude that $\phi$ assign less than $\opt_{\mmcs}(C)$ input wires to \textsc{True} but satisfies the circuit, which is a contradiction.

\section{Training a Single ReLU in the Realizable Case}
\label{sec:single_rel}

Here we demonstrate that training a single ReLU in the realizable case can be done in polynomial time using linear programming. The key observation is the following.
\begin{lemma}\label{lem:LP}
Consider a system $\Lambda$ of $m$ equalities of the form $[\langle{\bv_i},{\bx}\rangle]_+ = c_i$ where $\bv_i$ are fixed $n$-dimensional vectors and $\bx$ is an $n$-dimensional vector composed of the variables
$x_1, \ldots, x_n$. Then there is a polynomial time algorithm in $n,m$ and the binary representation of the numbers in $v_i,c_i$ to determine if $\Lambda$ is feasible, and, in the feasible case, output an assignment to the $x_i$'s satisfying all equalities in $\Lambda$.
\end{lemma}
\begin{proof}
We show how to transform each equality to a linear equality or inequality. Consider $[\langle{\bv_i},{\bx}\rangle]_+= c_i$. If $c_i<0$ then the inequality is not satisfied by any assignment and $\Lambda $ has no solution. If $c_i > 0$ then replace the equality by $\langle{\bv_i},{\bx}\rangle = c_i$. If $c_i=0$ then replace the equality by $\langle{\bv_i},{\bx}\rangle \leq 0$. Since transforming the equalities to linear (in)equalities can be done in polynomial time and as we can decide whether a system of linear inequalities over the reals is satisfiable in polynomial time using linear programing, the claimed statement follows.
\end{proof}

Recall a training sample $\{(\bx_i, y_i)\}_{i \in [m]}$ of a single ReLU is called \emph{realizable} if there exists a choice of weights $w_i,i \in [n]$ and a bias $b$
such that $[\langle {\bx_i},{\bw}\rangle+b]_+=y_i$ for all $i \in [m]$. Hence, the above lemma immediately implies that the training problem for a single ReLU can be solved in polynomial time for realizable samples.

\section{Missing proof of Proposition~\ref{prop:height-bound}}
\label{app:induction}


\begin{proof}[Proof of Proposition~\ref{prop:height-bound}]

Recall that we have the following constraints in our training sample:

{\bf Dummy Variable Constraint.} We add the following constraint
\begin{align} \label{eq:dummy-eps}
[w_\epsilon]_+ = \epsilon.
\end{align}



{\bf OR Gate Constraint.} For each OR gate with input wires $i_1, \dots, i_k$ and output wire $j$, we add the constraint
\begin{align} \label{eq:or-gate}
[w_j - w_{i_1} - \cdots - w_{i_k}]_+ = 0.
\end{align}

{\bf AND Gate Constraint.} For each AND gate with input wires $i_1, \dots, i_k$ and output wire $j$, we add the following $k$ constraints:
\begin{align} \label{eq:and-gate}
[w_j - w_{i_1}]_+ = 0, \cdots, [w_j - w_{i_k}]_+ = 0.
\end{align}

We will prove by induction on the height $h$.

{\bf Base Case.} Consider any input wire $i$ (of height 0) that is assigned \textsc{False} by $\phi$. By definition of $\phi$, we have $w_i < w_\epsilon$. Note that $w_\epsilon$ must be at most $\varepsilon + \sqrt{\delta}$, as otherwise the squared error incurred in~\eqref{eq:dummy-eps} is already more than $\delta$. Thus, we have $w_i \leq \varepsilon + \sqrt{\delta}$ as claimed.

{\bf Inductive Step.} Let $h \in \mathbb{N}$ and suppose that the statement holds for every \textsc{False} wire at height less than $h$. Let $j$ be any \textsc{False} at height $h$. Let us consider two cases:
\begin{itemize}
\item $j$ is an output of an OR gate. Let $i_1, \dots, i_k$ be the inputs of the gate. Since $j$ is evaluated to \textsc{False}, $i_1, \dots, i_k$ must all be evaluated to \textsc{False}. From our inductive hypothesis, we have $w_{i_1}, \dots, w_{i_k} \leq (2|C|)^{h - 1} \cdot (\varepsilon + \sqrt{\delta})$. Now, observe that $w_j$ can be at most $\sqrt{\delta} + w_{i_1} + \cdots + w_{i_k}$, as otherwise the squared error incurred in~\eqref{eq:or-gate} would be more than $\delta$. As a result, we have
\begin{align*}
w_j &\leq \sqrt{\delta} + k \cdot (2|C|)^{h - 1} \cdot (\varepsilon + \sqrt{\delta}) \\
&= \sqrt{\delta} + |C| \cdot (2|C|)^{h - 1} \cdot (\varepsilon + \sqrt{\delta}) \\
&\leq (2|C|)^h \cdot(\varepsilon + \sqrt{\delta}).
\end{align*}
\item $j$ is an output of an AND gate. Let $i_1, \dots, i_k$ be the inputs of the gate. Since $j$ is evaluated to \textsc{False}, at least one of $i_1, \dots, i_k$ must all be evaluated to \textsc{False}. Let $i$ be one such wire. Observe that $w_j$ can be at most $\sqrt{\delta} + w_i$, as otherwise the squared error incurred in~\eqref{eq:and-gate} would be more than $\delta$. Hence, we have
\begin{align*}
w_j &\leq \sqrt{\delta} + w_i \\
&\leq \sqrt{\delta} + (2|C|)^{h - 1} \cdot (\varepsilon + \sqrt{\delta}) \\
&\leq (2|C|)^h \cdot (\varepsilon + \sqrt{\delta}).
\end{align*}
where the second inequality comes from the inductive hypothesis.
\end{itemize}
In both cases, we have $w_j < (2|C|)^h  \cdot (\varepsilon + \sqrt{\delta})$, which concludes the proof of Proposition~\ref{prop:height-bound}.
\end{proof}

\section{The Running Time of Arora et al.'s Algorithm} \label{sec:exp-algo}

\cite{arora2018understanding} gives a simple algorithm that runs in time $n^{O(km)}$ and outputs the optimal training error (to within arbitrarily small accuracy). Below, we observe that their algorithm also yields an $2^{km} \cdot poly(n, m, k)$ time algorithm; we use this running time guarantee for agnostically learning depth-2 networks of ReLUs. Before we proceed to the statement and the proof of the algorithm, we remark that, our NP-hardness proof for 2 ReLUs in fact also implies that, assuming the Exponential Time Hypothesis (ETH)~\cite{IP01,IPZ01}\footnote{ETH states that 3SAT with $n$ variables and $m = O(n)$ clauses cannot be solved in $2^{o(n)}$ time.}, the training problem for 2 ReLUs cannot be done in $2^{o(m)}$ time. Hence, the dependency $m$ in the exponent is tight in this sense.

\begin{lemma} \label{lem:exp-algo-app}
There is an $2^{km} \cdot poly(n, m,1/\delta,C)$-time algorithm that, given samples $\{(\bx_i, y_i)\}_{i \in [m]}$ where $\bx_i \in \mathbb{R}^n$ and an accuracy parameter $\delta \in (0,1)$, finds the weights $\bw_1, \dots, \bw_k \in \cB^n$ and biases $b_1, \dots, b_k \in [-1, 1]$ that minimizes the function $$\sum_{i \in [m]} \left(y_i - \sum_{j \in [k]} [\left<\bw_j, \bx_i\right> + b_j]_+\right)^2$$ up to an additive error of $\delta$. We assume the bit complexity of every number appearing  in the coordinates of the $x_i$'s and $y_i$'s is at most $C$. Furthermore, there is an algorithm with the same running time up to polynomial factors that finds $\bw_1, \dots, \bw_k \in \cB^n, b_1, \dots, b_k \in [-1, 1]$ subjects to an additional constraint that $\sum_{j \in [k]} [\left<\bw_j, \bx_j\right> + b_j]_+ = 0$ for all $i$ such that $y_i = 0$.

Moreover, these problems can be solved with similar running time even when we require additional constraints that $\bw_1, \dots, \bw_k \in \cB^n := \{\bw \in \mathbb{R}^n \mid \|\bw\|_2 \leq 1\}$ and $b_1, \dots, b_k \in [-1, 1]$.
\end{lemma}

\begin{proof}
For each ReLU term $[\left<\bw_j, \bx_i\right> + b_j]_+$ guess whether it equals $0$ or $\left<\bw_j, \bx_i\right> + b_j$ and replace the term in the error function accordingly.
Furthermore, if the guess $[\left<\bw_j, \bx_i\right> + b_j]_+=0$ was made then add the linear constraint $\left<\bw_j, \bx_i\right> + b_j\leq 0$. Else, add the
linear constraint $\left<\bw_j, \bx_i\right> + b_j\geq 0$. Finally add the constraints $-1 \leq b_i \leq 1, \|\bw_i\|_2 \leq 1$ for all $1 \leq i \leq k$.
After all guesses are made we get a convex quadratically constrained quadratic
program (QCQP). It is well known that such a convex optimization problem can be solved in time polynomial in
$n,m,1/\delta,C$ using a separation oracles and the ellipsoid algorithm (see for example, \cite{bubeck2015convex}, section 2.1). Since the number of guesses is at most $(2^m)^k$, the claim follows.
For the second part of the lemma, simply substitute the constraint $\sum_{j \in [k]} [\left<\bw_j, \bx_i\right> + b_j]_+ = 0$ according to the guesses made and add the resulting linear constraint. The claim follows.
\end{proof}

\end{document}